\documentclass[12pt,reqno]{amsart}
\usepackage{amsmath, amsthm, amssymb, cite,color}
\usepackage{graphicx}
\usepackage{enumerate}
\newtheorem{theorem}{Theorem}
\newtheorem{lemma}[theorem]{Lemma}\newtheorem{corollary}[theorem]{Corollary}

\newtheorem{prop}[theorem]{Proposition}

\newtheorem{definition}[theorem]{Definition}

\begin{document}

\title[Graphyne spectrum]{Quantum graph spectra of a graphyne structure}
\author{Ngoc T. Do and Peter Kuchment}
\address{Mathematics Department, Texas A\&M University, College Station, TX 77843-3365}
\email{dothanh@math.tamu.edu, kuchment@math.tamu.edu}
\date{}
\maketitle
\begin{abstract}
We study the dispersion relations and spectra of invariant Schr\"odinger operators on a graphyne structure (lithographite).
In particular, description of different parts of the spectrum, band-gap structure, and Dirac points are provided.
\end{abstract}

\section*{Introduction}
Graphene, a monolayer of graphite, is famous for its unusual electric and mechanical properties (e.g., \cite{Novoselov,Geim}).
Recently, researchers suggested other 2D carbon allotropes which were given the common name "graphynes."
It has been suggested (see, e.g. \cite{Gorling,Bardhan}) that some graphynes, which have not been synthesized yet, might be even more interesting than the graphene.

Various standard and less standard approaches have been used to model the spectral structure of graphene and graphynes (one of the most popular was to use a version of density functional technique \cite{Ivanovskii}). One of the ways similar 2D structures have been modeled previously, was using the techniques of quantum networks, also known as quantum graphs (see, e.g., \cite{Amovilli,Ruedenberg}). In particular, several studies of spectra of Schr\"{o}dinger operators on graphene and carbon nanotube structures (e.g., \cite{Korotyaev1,Korotyaev2,Kuch_Post}) have been conducted, which have proven to be much simpler to study and preserving all essential ingredients of the dispersion relation. One should also be aware of a recent study of dispersion relations of 2D Schr\"odinger operators with honeycomb symmetry  in \cite{Fefferman,Fefferman_arxiv}, where in particular the mandatory presence of Dirac cones is established.

In this paper, we take the quantum graph approach similar to \cite{Kuch_Post} to study spectra of Schr\"odinger operators on the simplest graphyne among 14 various structures suggested in \cite{Ivanovskii} (it represents the 2D projection of the so called lithographite \cite{BucCas}).
From now on we reserve the word "graphyne" for this particular structure.
We derive the dispersion relations for these operators on graphyne.
From here, we extract various information about the spectral structure of the operators.
Unlike similar periodic operators in $\mathbb{R}^n$, the quantum graph operators can (and often do) have point spectrum (i.e., bound states).
We find this part of the spectrum and provide an explicit description of the corresponding eigenspaces\footnote{The presence of bound states is an artefact of the quasi-1D model. It, however, often indicates possible presence of very flat bands in the "grown up" system.}.
The presence of spectral gaps and conical ``Dirac'' points is also studied.
The formulations of the results involve the discriminant of the Hill operator with the potential obtained by periodic extension of the 1D potential on a single edge.

In Section \ref{S:geometry} we introduce the geometry of the structure and the operators of interest.
In Section \ref{S:spectrumgraphyne} we derive the dispersion relation and the band-gap structure for graphyne with the main results stated in Theorems \ref{T:main_graphyne} and \ref{T:conical_singularity}.
The proof of an auxiliary Proposition \ref{L:F_property} is given in Section \ref{S:Prop}.

\section{Geometry of graphyne structure and related Schr\"odinger operators}\label{S:geometry}

The graphyne structure that we study is represented by the graph $G$ shown in Fig. \ref{F:G} and has square symmetry, unlike the graphene's honeycomb one.
At each vertex there is a carbon atom that is bonded to other atoms. Chemical bonds between atoms are represented by the
edges connecting the corresponding vertices.
\begin{center}
\begin{figure}[ht!]
\includegraphics[scale=0.4]{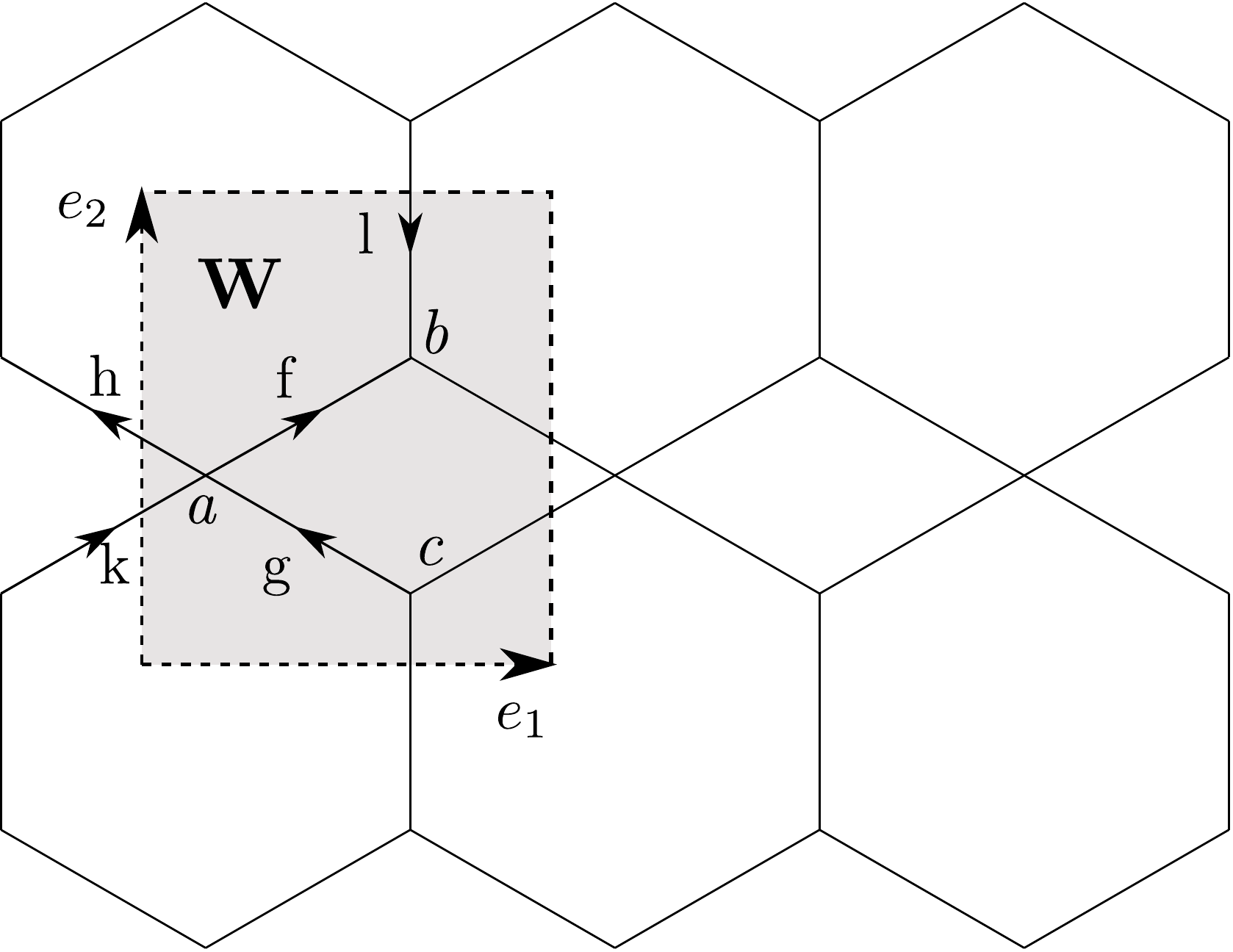}
\caption{The structure $G$ and a fundamental domain $W$ with vertices $a,b,c$ and parts of edges $f,g,h,k,l$ (and some of their lattice shifts)}
\end{figure}\label{F:G}
\end{center}
All the edges of $G$ are assumed to have length 1. We denote by $E(G)$ and $V(G)$ correspondingly the set of all edges and all vertices of $G$.
There is a free action of the group $\mathbb{Z}^2$ of integer vectors in $\mathbb{R}^2$ on G by the shifts by vectors
$p_1 e_1 + p_2 e_2,$ where $(p_1, p_2) \in \mathbb{Z}^2$ and $e_1=(\sqrt{3},0), e_2=(0,2)$.
We choose the domain $W$ shown in Fig. \ref{F:G} as the fundamental domain of this action.
It contains three vertices $a, b, c$ and pieces of five edges $f, g, h , k, l$ as shown in the figure.
We choose the directions of these edges as shown. Notice that the boundary of the chosen fundamental domain does not contain  any vertices (which is always possible to achieve). The reason for using such domain is that the presence of a vertex on its boundary would unnecessarily complicate considerations.
The entire structure $G$ can be obtained from $W$ by $\mathbb{Z}^2$-shifts, which also define directions on all edges of the graph $G$.

We define on each directed edge $e$ the arc length coordinate $x_e$ that identifies it with the directed segment $[0,1]$.
When it does not lead to ambiguity, we will use $x$ instead of $x_e$ to denote the coordinate on the edge $e$.
One can introduce now in a natural way the Hilbert space $L_2(e)$ as the space of all square integrable functions on the edge $e$ and $H_2(e)$ as the Sobolev space on $e$ that consists of functions with two distributional derivatives in $L_2(e)$.

We also define
$$
L_2(G)=\bigoplus_{e \in E(G)} L_2(e)
$$
as the space of all square integrable functions on $G.$

We will use the notations $u_e$ for the restriction of a function $u$ on $G$ to an edge $e$. We also use $u_e'$ for the derivative of $u_e$ in the direction of  the edge $e$.

Let $q_0(x)$ be an even and real $L_2$-function on $[0,1]$, i.e. $q_0(x) = q_0(1-x)$ for a.e. $x\in [0,1]$.

Using the described before identification of the directed edges with the segment $[0,1]$, we can transfer the potential $q_0(x)$ to each edge,
thus defining a potential $q(x)$ on the whole $G$.

It is not hard to show that the evenness assumption on $q_0$ implies the following property:
\begin{lemma}
The potential $q$ defined as above is invariant with respect to the symmetry group of the graph $G$.
\end{lemma}
We are now ready to construct the Schr\"odinger operator $H$ in $L_2(G)$, whose spectral properties will be studied in this paper.
The operator $H$ acts on each edge as follows:
\begin{equation}
Hu(x)=-\frac{d^2 u(x)}{dx^2}+ q(x) u(x).
\label{E:H}
\end{equation}
Its domain $D(H)$ consists of all functions $u(x)$ on $G$ such that:
\begin{enumerate}
 \item $u_e \in H_2(e)$, for all $e \in E(G),$
 \item
   \begin{equation}
   \sum_{e\in E(G)}\|u_e\|^2_{H_2(e)}<\infty,
   \label{E:D(H)_condition}
   \end{equation}
 \item at each vertex these functions satisfy \emph{Neumann vertex condition}, i.e. $u_{e_1}(v)=u_{e_2}(v)$ for any edges $e_1,e_2$ containing the vertex $v$ and
   $$
   \sum_{v\in e} u_e' (v)=0, \text{ for any vertex } v \text{ in } V(G).
   $$
\end{enumerate}

Thus defined operator is well-known to be unbounded and self-adjoint (e.g. \cite[Theorem 1.4.19]{Berk_Kuch}).
It is also invariant with respect to all symmetries of the graph $G$.

\section{Graphyne spectrum}\label{S:spectrumgraphyne}

In this section we study the spectrum of the operator $H$.
Let us describe first the main steps of our approach.
The technique of Floquet-Bloch theory \cite{Reed_Simon_4, Eastham,Kuch_Floquet_book,Berk_Kuch} allows us to reduce the consideration to a family of spectral problems on the fundamental domain $W$.
Then one can switch to a discrete problem (e.g., \cite{Pankr}, \cite[Section 3.6]{Berk_Kuch}).
This step uses the standard Hill's operator theory \cite{Eastham,Reed_Simon_4}.
Finally, the discrete problem can be analyzed rather explicitly.

Let us get to some detail now.
For each $\theta=(\theta_1, \theta_2)$ in the \emph{Brillouin zone} $B=[-\pi,\pi]^2$, let $H^\theta$ be the \emph{Bloch Hamiltonian} that acts as (\ref{E:H})
on the domain that consists of functions $u(x)$ that belong to $H^2_{loc}(G)$ and satisfy Neumann vertex condition along with the following \emph{cyclic} (or \emph{Floquet}) condition:
\begin{equation}
\label{E:cyclic_condition}
u(x+p_1 e_1+p_2 e_2)=u(x)e^{ip\theta}=u(x)e^{i(p_1 \theta_1+p_2 \theta_2)},
\end{equation}
for all $(p_1,p_2)\in \mathbb{Z}^2$ and all $x\in G.$

Due to this condition, such functions $u$ are uniquely determined by their restrictions to the fundamental domain $W.$

We have the direct integral expansion \cite[Section 4.3]{Berk_Kuch}
$$H=\int_B^\oplus H^\theta d\theta.$$
Therefore \cite[Theorem 4.3.1]{Berk_Kuch},
$$
\sigma(H)=\bigcup_{\theta\in [-\pi,\pi]^2} \sigma(H^\theta).
$$
It is well-known (e.g., \cite{Kuch_Quantum1, Berk_Kuch}) that operator $H^\theta$ has purely discrete spectrum $\sigma(H^\theta)=\{\lambda_j(\theta)\}$ with
$\displaystyle \lim_{j\rightarrow\infty}\lambda_j(\theta)=\infty$.
The multiple-valued function $\theta \longmapsto \{\lambda_j(\theta)\} $ is called \emph{dispersion relation} and its graph - \emph{dispersion surface} or \emph{Bloch variety} of the operator $H$.
Spectrum of $H$ is the range of the dispersion relation for $\theta$ changing in the Brillouin zone.
Thus, we now concentrate on studying the spectrum of $H^\theta,$ for $\theta\in B$, i.e on solving the eigenvalue problem:
\begin{equation}
\qquad H^\theta u =\lambda u, \lambda\in\mathbb{R},
\label{E:H_theta_equation}
\end{equation}
for $u\in H_2(W)$ satisfying the cyclic condition (\ref{E:cyclic_condition}) at the boundary and Neumann vertex condition inside.

Combining vertex and cyclic conditions we have:
$$
\left\{\begin{array}{lr}
u_f(0)=u_g(1)=u_h(0)=u_k(1)=:A\\
u_f '(0)-u_g '(1)+ u_h'(0)-u_k'(1)=0\\
u_f(1)=u_l(1)=u_h(1)e^{i\theta_1}=:B\\
u_f'(1)+u_l'(1)+u_h'(1)e^{i\theta_1}=0\\
u_g(0)=u_k(0)e^{i\theta_1}=u_l(0)e^{-i\theta_2}=:C\\
u_g'(0)+u_k'(0)e^{i\theta_1}+u_l'(0)e^{-i\theta_2}=0.
\end{array}\right.
$$

We will need another auxiliary operator:
\begin{definition}
We denote by  $H^D$ the \emph{Dirichlet Hamiltonian} on $[0,1]$ that acts  as (\ref{E:H}) with Dirichlet boundary conditions $u(0)=u(1)=0$. We also denote by $\Sigma^D$ the (discrete) spectrum of  $H^D$.
\end{definition}
For each $\lambda\notin \Sigma^D,$ there exist two linearly independent solutions $\varphi_0,\varphi_1$ such that $\varphi_0(0)=\varphi_1(1)=1, \varphi_0(1)=\varphi_1(0)=0.$
Sometimes we address $\varphi_0, \varphi_1$ as $\varphi_{0,\lambda}$ and $\varphi_{1,\lambda}$ to emphasize their dependence on $\lambda.$
We use the same notation $\varphi_0, \varphi_1$ for analogous functions on each edge of $W$ under fixed identification of these edges with the segment $[0,1]$, which should not lead to a confusion.
Then for $\lambda\notin\Sigma^D$ solution of (\ref{E:H_theta_equation}) can be represented as follows:
$$\left\{\begin{array}{lr}
u_f=A\varphi_0+B\varphi_1\\
u_g=C\varphi_0+A\varphi_1\\
u_k=Ce^{-i\theta_1}\varphi_0+A\varphi_1\\
u_h=A\varphi_0+Be^{-i\theta_1}\varphi_1\\
u_l=Ce^{i\theta_2}\varphi_0+B\varphi_1.
\end{array}\right.$$
Continuity and eigenvalue equation on each edge are already satisfied.
What is left to be checked is the zero flux condition at each of the three vertices in $W$:
\begin{equation}
\left\{\begin{array}{lr}
A(2\varphi_0'(0)-2\varphi_1'(1))+(B\varphi_1'(0)-C\varphi_0'(1))(1+e^{-i\theta_1})=0\\
A\varphi_0'(1)(1+e^{i\theta_1})+3B\varphi_1'(1)+Ce^{i\theta_2}\varphi_0'(1)=0\\
A\varphi_1'(0)(1+e^{i\theta_1})+Be^{-i\theta_2}\varphi_1'(0)+3C\varphi_0'(0)=0.
\end{array}\right.
\label{E:zero_flux_cond}
\end{equation}
Notice that $\varphi_1'(1)=-\varphi_0'(0)$ and $\varphi_1'(0)=-\varphi_0'(1)$
due to the evenness of function $q_0.$
Thus (\ref{E:zero_flux_cond}) becomes
\begin{equation}
\left\{\begin{array}{lr}
-4A\varphi_1'(1)+B\varphi_1'(0)(1+e^{-i\theta_1})+C\varphi_1'(0)(1+e^{-i\theta_1})=0\\
A\varphi_1'(0)(1+e^{i\theta_1})-3B\varphi_1'(1)+Ce^{i\theta_2}\varphi_1'(0)=0\\
A\varphi_1'(0)(1+e^{i\theta_1})+Be^{-i\theta_2}\varphi_1'(0)-3C\varphi_1'(1)=0.
\end{array}\right.
\label{E:ABC_varphi}
\end{equation}

Since $\varphi_1'(0)\neq 0,$ we can define
\begin{equation}
\eta(\lambda):=\frac{\varphi_{1,\lambda}'(1)}{\varphi_{1,\lambda}'(0)};
\label{E:eta}
\end{equation}
then (\ref{E:ABC_varphi}) is reduced to

$$\left\{\begin{array}{lr}
-4\eta(\lambda)A+(1+e^{-i\theta_1})B+(1+e^{-i\theta_1})C=0\\
(1+e^{i\theta_1})A-3\eta(\lambda)B+ e^{i\theta_2}C=0\\
(1+e^{i\theta_1})A+e^{-i\theta_2}B-3\eta(\lambda)C=0.
\end{array}\right.$$

Determinant of this system is $$-4[9\eta^3(\lambda)-\eta(\lambda)-(\cos\theta_1+1)(3\eta(\lambda)+\cos\theta_2)].$$
These calculations prove the following:

\begin{lemma}\label{L:rootfunction}
A point $\lambda\notin\Sigma^D$ is in the spectrum of the Schr\"odinger operator $H$ if and only if there exists $\theta=(\theta_1,\theta_2)\in B$ such that $x=\eta(\lambda)$ is
a root of the equation
\begin{equation}
\qquad 9x^3-x-(\cos\theta_1+1)(3x+ \cos\theta_2)=0.
\label{E:eta_equation}
\end{equation}
\end{lemma}

Let us now extend $q_0$ periodically from $[0,1]$ to the whole real axis $\mathbb{R}$ and consider the Hill operator $H^{per}$ on $\mathbb{R}$ as below:
$$\displaystyle H^{per}u(x)=-\frac{d^2u(x)}{dx^2}+ q_0(x)u(x).$$
Here we use the same notation $q_0(x)$ for the periodic extension.
The monodromy matrix $M(\lambda)$ of $H^{per}$ is defined by the formula
$$\begin{bmatrix}\varphi(1)\\ \varphi'(1)\end{bmatrix}= M(\lambda)\begin{bmatrix}\varphi(0)\\ \varphi'(0)\end{bmatrix},$$
where $\varphi$ satisfies the differential equation
\begin{equation}
\qquad -\frac{d^2\varphi(x)}{dx^2}+ q_0(x)\varphi(x)=\lambda \varphi(x) \text{ on } \mathbb{R}.
\label{E:Hill_equation}
\end{equation}
\emph{Discriminant} (or \emph{Lyapunov function}) $tr M(\lambda)$ of the Hill operator $H^{per}$ is denoted by $D(\lambda).$
Next proposition (\cite[Proposition 3.4]{Kuch_Post}) collects some well-known results about the spectra of Hill operators \cite{Eastham}:
\begin{lemma}\label{L:KuchPost}
\indent
\begin{enumerate}
     \item The spectrum $\sigma(H^{per})$ of $H^{per}$ is purely absolutely continuous.

     \item $\sigma(H^{per})=\{\lambda\in\mathbb{R}\big||D(\lambda)|\leq 2\}$.

     \item $\sigma(H^{per})$ consists of the union of closed non-overlapping (although, possibly touching) and non-zero length finite intervals (bands)
     $B_{2k}:=[a_{2k},b_{2k}], B_{2k+1}:=[b_{2k+1},a_{2k+1}]$ such that
     $$a_0<b_0\leq b_1<a_1\leq a_2<b_2\leq\ldots$$
     and $\displaystyle\lim_{k\rightarrow\infty}a_k=\infty.$
     \\
     The (possibly empty) segments $(b_{2k},b_{2k+1})$ and $(a_{2k},a_{2k+1})$ are called the spectral gaps.
     \\
     Here $\{a_k\}$ and $\{b_k\}$ are the spectra of the operators with periodic and anti-periodic conditions on $[0,1]$ correspondingly.
     \item Let $\lambda_k^D\in\Sigma^D$ be the $k^{th}$ Dirichlet eigenvalue labeled in increasing order.
     Then $\lambda_k^D$ belongs to (the closure of) the $k^{th}$ gap. When $q_0$ is even, $\lambda_k^D$ coincides with
     an edge of the $k$-th gap.
     \item If $\lambda$ is in the interior of the $k^{th}$ band $B_k$, then $D'(\lambda)\neq 0,$ and $D(\lambda)$ is a homeomorphism of the band $B_k$
     onto $[-2,2].$
     Moreover, $D(\lambda)$ is decreasing on $(-\infty,b_0)$ and $(a_{2k},b_{2k})$ and is increasing on $(b_{2k+1},a_{2k+1})$.
     It has a simple extremum in each spectral gap $[a_k,a_{k+1}]$ and $[b_k,b_{k+1}].$
     \item The dispersion relation for $H^{per}$ is given by
     $$D(\lambda)=2\cos\theta,$$
     where $\theta$ is the one-dimensional quasimomentum.
\end{enumerate}
\end{lemma}
Claim (4) of the lemma about the even potential case can be explained as follows:
let $u(x)$ be the eigenfunction of the Hill operator $H^{per}$ corresponding to the $k^{th}$ Dirichlet eigenvalue $\lambda_k^D,$ i.e.
$$H^{per} u(x)=\lambda_k^D u(x), u(0)=u(1)=0.$$
Then $u(1-x)$ is also an eigenfunction corresponding to that eigenvalue.
Either $u(x)+ u(1-x)$ or $u(x)-u(1-x)$ is nonzero and therefore will be an eigenfunction corresponding to $\lambda_k^D$.
Since $u(x)+u(1-x)$ is periodic and $u(x)-u(1-x)$ is anti-periodic, $\lambda_k^D$ must coincide with an edge of the $k^{th}$ gap.

The following relation between function $\eta(\lambda)$ (see (\ref{E:eta})) and the discriminant $D(\lambda)$ of $H^{per}$ is well-known \cite{Eastham, Kuch_Post} and easy to establish:
\begin{equation}
\eta(\lambda)=\frac{1}{2}D(\lambda).
\label{E:eta_discriminant}
\end{equation}
Therefore, according to the statement (2) of Lemma \ref{L:KuchPost} and equality (\ref{E:eta_discriminant}), one needs to analyze the roots of the cubic equation (\ref{E:eta_equation}).

So far we have just dealt with $\lambda$ not in the Dirichlet spectrum $\Sigma^D$ of $H^{D}$.
Now let us consider the exceptional values $\lambda\in\Sigma^D$.

We introduce the following notion first:
\begin{definition}\label{D:loop}
An eigenfunction is said to be a \emph{simple loop state} if it is supported on a single hexagon or rhombus of the structure $G$ and vanishes at all vertices (see Fig. \ref{F:loopstate}).
\end{definition}

We can now describe what happens for $\lambda\in\Sigma^D$.
\begin{lemma}\label{L:simple_loop_state}
Each $\lambda\in\Sigma^D$ is an eigenvalue of infinite multiplicity of the operator $H$.
The corresponding eigenspace is generated by\footnote{I.e., is the closed linear hall of ...} the simple loop states.
\end{lemma}

\itshape{Proof.}
\upshape
For each $\lambda\in\Sigma^D,$ let $\psi_{\lambda}$ be the corresponding eigenfunction of operator $H^D$.
Since $q_0$ is even, one can assume function $\psi_\lambda$ to be either even or odd.
For an odd eigenfunction $\psi_\lambda$, we repeat it on each edge of hexagon/rhombus;
for even eigenfunction $\psi_\lambda$ - repeat around hexagon/rhombus with an alternating sign.
In both cases we get an eigenfunction of $H$ that lives only on one particular loop.
Thus $\lambda\in\sigma_{pp}(H)$.
\begin{figure}[ht!]
\includegraphics[scale=0.6]{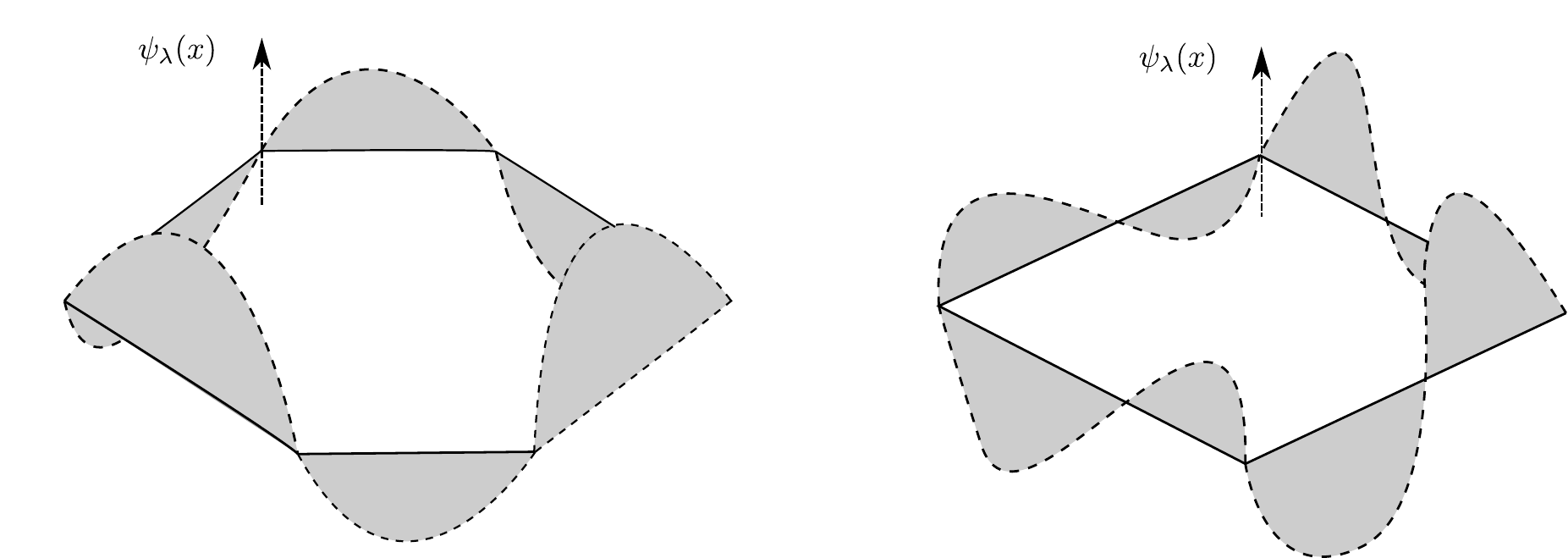}
\caption{Simple loop states constructed from even function on $[0,1]$ for hexagon and odd function on $[0,1]$ for rhombus}\label{F:loopstate}
\end{figure}
Let us prove infinite multiplicity of the eigenvalues $\lambda\in\Sigma^D$, which is a well-known feature of periodic problems.
Let $M_{\lambda}\subset L_2(G)$ be the corresponding eigenspace and $\gamma$ be some period vector of $G$.
The shift operator $S_{\gamma}$ by $\gamma$ acts in $M_\lambda$ as an unitary operator.
Suppose that $M_{\lambda}$ is finite dimensional, then $S_{\gamma}$ has an eigenfunction $f\in M_{\lambda}\subset L_2(G)$ corresponding to an eigenvalue $\mu$ such that $|\mu|=1.$
On the other hand, $f$ is multiplied by $\mu$ when shifted by the vector $\gamma$. Since $|\mu|=1$, $f$ clearly cannot belong to $L_2(G)$, which leads to contradiction.
Thus $\lambda\in\Sigma^D$ are eigenvalues of infinite multiplicity.

In order to prove that the eigenspace is generated by simple loop states with hexagonal and rhomboidal supports, it is enough to prove that these simple loop states generate all compactly supported eigenfunctions in the eigenspace $M_\lambda.$
Indeed, as it is shown in \cite{Kuch_Quantum2} (see also \cite[Theorem 4.5.2]{Berk_Kuch}), linear combinations of compactly supported eigenfunctions are dense in the space $M_\lambda$.

First we notice that each compactly supported eigenfunction $\varphi$ of $H$ vanishes at all vertices.
Indeed, due to connectedness of $G$, there must be a ``boundary'' vertex $v$ of the support that is connected by an edge with a vertex $w$ outside the support.
We claim that  $\varphi (v)=0$.
Otherwise, we have an edge such that the function vanishes at one end, $w$ (corresponding to $x=0$) and does not vanish at the other end.
We introduce a basis of solutions of (\ref{E:H_theta_equation}), functions $c_\lambda$ and $s_\lambda,$ such that $c_\lambda(0)=1,s_\lambda(0)=s_\lambda(1)=0$ (a non-trivial function $s_\lambda$ exists, since $\lambda\in\Sigma^D$).
The eigenfunction $\varphi$ can be represented as $\varphi(x)=A c_\lambda(x)+B s_\lambda(x)$.
In particular, $0=\varphi(0)=A,$ and so $0\neq\varphi(1)=B s_\lambda(1)=0,$ which leads to contradiction.
Repeating this argument, we conclude that the eigenfunction $\varphi$ vanishes at all vertices.
Besides, the support of $\varphi$ cannot have a vertex of degree 1. (Otherwise, due to Neumann boundary condition, both function and its derivative will vanish at that vertex, which makes function to be equal to zero.)

Now one needs to prove that $\varphi$ can be represented as a combination of simple loop states.
Consider the external boundary of the support of $\varphi$, which is a closed circuit $C$ of edges, containing the whole support inside.
The interior of this curve is a union of $N$ elementary hexagons and/or rhombuses of the graph $G$.
We begin with a boundary edge $e_0\in C$.
One of the $N$ internal hexagonal or rhomboidal loops must contain $e_0$.
Let $\varphi_0$ be the simple loop state that coincides with $\varphi$ on the edge $e_0$ and is extended to that loop as described before.
Function $\varphi-\varphi_0$ will be the new eigenfunction with a smaller support (number of loops $N-1$).
Continuing this process (see Fig.\ref{F:loop}), we will eventually represent the eigenfunction $\varphi$ as a combination of simple loop states. \qed
\begin{figure}[ht!]
\includegraphics[scale=0.4]{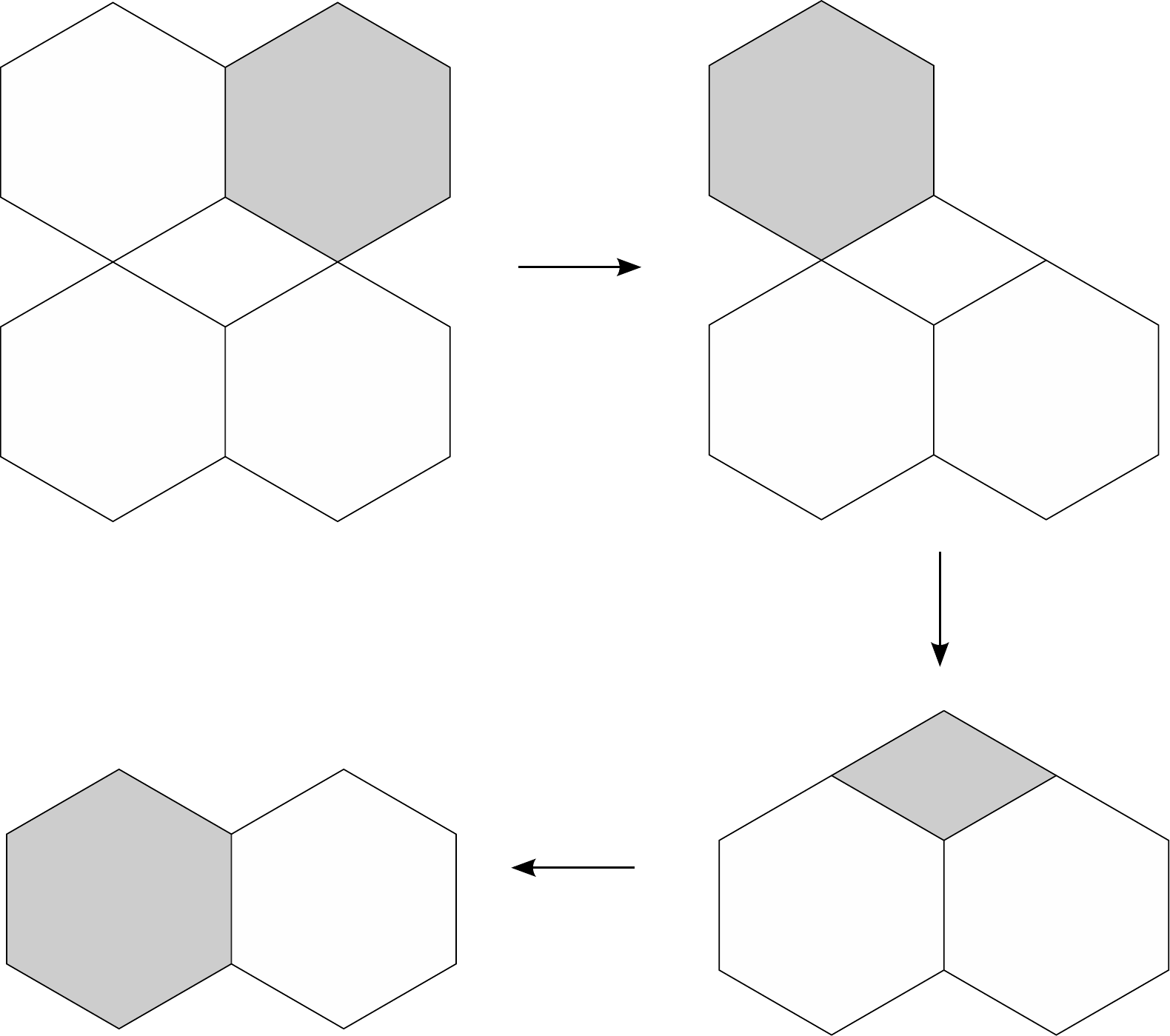}
\caption{An example of deleting simple loop states (the dark ones) from the support of an eigenfunction}\label{F:loop}
\end{figure}

In the next theorem, which is our main result of this section, we describe the dispersion relation and the structure of the spectrum of operator $H$.

Let $F(\theta)$ be the triple-valued function providing for each $\theta$ the three roots of the equation (\ref{E:eta_equation}).
By Proposition \ref{L:F_property}, which we will formulate and prove later, function $F$ is real-valued in the Brillouin zone.
Assume $F(\theta)=(F_1(\theta),F_2(\theta),F_3(\theta))$, where $F_1(\theta)\leq F_2(\theta)\leq F_3(\theta)$ for all $\theta\in B$.
Then we have

\begin{theorem}
\indent
\begin{enumerate}
     \item The singular continuous spectrum $\sigma_{sc}(H)$ is empty.

     \item The dispersion relation of operator $H$ consists of the following two parts:\\
     i) pairs $(\theta,\lambda)$ such that $0.5D(\lambda)\in F(\theta)$ (or, $\lambda\in D^{-1}(2F(\theta))$), where $\theta$ changing in the Brillouin zone;\\
     and \\
     ii) the collection of flat (i.e., $\theta$-independent) branches $\lambda\in\Sigma^D$.

     \item The absolutely continuous spectrum $\sigma_{ac}(H)$ has band-gap structure and is (as the set) the same as the spectrum $\sigma(H^{per})$
     of the Hill operator $H^{per}$ with potential obtained by extending periodically $q_0$ from $[0,1]$.
     In particular, $$\sigma_{ac}(H)=\{\lambda\in\mathbb{R}\big|\left|D(\lambda)\right|\leq 2\},$$
     where $D(\lambda)$ is the discriminant of $H^{per}.$

     \item The bands of $\sigma(H)$ do not overlap (but can touch).
     Each band of $\sigma(H^{per})$ consists of three touching bands of $\sigma(H)$.

     \item The pure point spectrum $\sigma_{pp}(H)$ coincides with $\Sigma^D$ and belongs to the union of the edges of spectral gaps of
     $\sigma(H^{per})=\sigma_{ac}(H)$.
     \\
     Eigenvalues $\lambda\in\Sigma^D$ of the pure point spectrum are of infinite multiplicity and the corresponding eigenspaces are
     generated by simple loop (hexagon or rhombus) states.

     \item Spectrum $\sigma(H)$ has gaps if and only if $\sigma(H^{per})$ has gaps.
\end{enumerate}
\label{T:main_graphyne}
\end{theorem}
The statements of the theorem are illustrated in Fig. \ref{F:triple}.
\begin{figure}[ht!]
\includegraphics[scale=0.8]{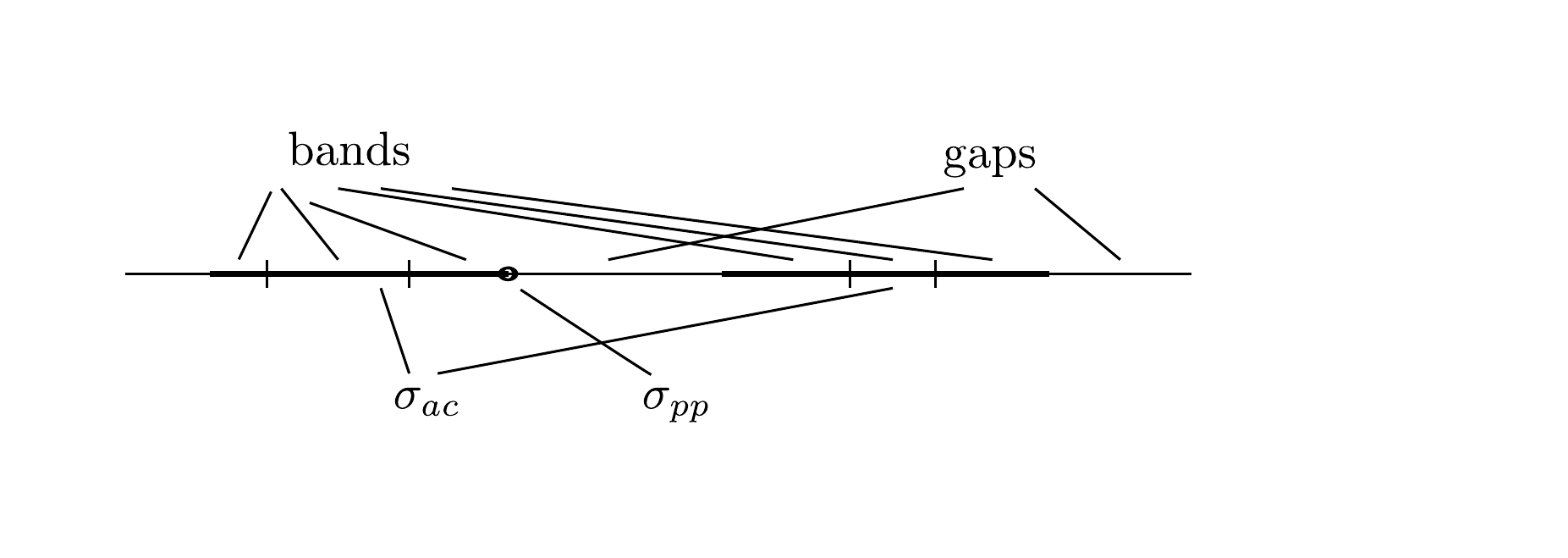}
\caption{The bold segments are the bands of $\sigma(H^{per})$. Each of them is split into three touching bands of $\sigma(H)$. One eigenvalue at the end of a band is also shown.}\label{F:triple}
\end{figure}
\begin{proof}
The first claim about emptiness of the singular continuous spectrum of Schr\"odinger operator $H$ is rather well-known\footnote{It goes back to the famous L.~Thomas' absolute continuity theorem \cite{Thomas}. See, e.g., \cite[Theorem 4.4.1]{Berk_Kuch}, or \cite{Reed_Simon_4} and references therein.}.

Let first $\lambda\notin\Sigma^D$.
Then, according to Lemma \ref{L:rootfunction}, $(\theta,\lambda)$ is in the dispersion surface of $H$ iff $\eta(\lambda)$ is a root of (\ref{E:eta_equation}) for this $\theta$.
In other words, due to (\ref{E:eta_discriminant}), $0.5D(\lambda)$ must be one of the three values of $F(\theta)$.
Hence, this is equivalent to $D(\lambda)\in 2F(\theta)$, or $\lambda\in D^{-1}(2F(\theta))$.

If $\lambda\in\Sigma^D$, then according to Lemma \ref{L:simple_loop_state}, $(\theta,\lambda)$ is in the dispersion surface for any $\theta$ from the Brillouin zone.

This proves the second statement of the theorem.

According to the Lemmas \ref{L:KuchPost} and \ref{L:simple_loop_state} we have
\begin{equation} \label{E:Dirichlet_spectrum}
 \Sigma^D\subset\sigma(H), \Sigma^D\subset\sigma(H^{per})
\end{equation}
and
$$ \sigma(H^{per})=\{\lambda\in\mathbb{R}\big|\left|D(\lambda)\right|\leq 2\}.$$
For $\lambda\notin\Sigma^D$, $\lambda$ is in the spectrum of $H$ iff $\eta(\lambda)$ is a root of equation (\ref{E:eta_equation}) for some $\theta$. Proposition  \ref{L:F_property} below shows, in particular, that all roots of equation (\ref{E:eta_equation}) belong to $[-1,1]$ and cover this interval. Thus, $\lambda$ is in the spectrum of $H$ iff $|\eta(\lambda)|\leq 1$.
Since $D(\lambda)=2\eta(\lambda)$, this means that $\sigma(H)\setminus \Sigma^D=\sigma(H^{per})\setminus \Sigma^D$ and, by closure,
$\sigma(H)=\sigma(H^{per})$.

The same Proposition \ref{L:F_property} shows that the graph of the triple-valued function $F(\theta)$ does not have any flat branches outside the set of values $\Sigma^D$.
Thus, the spectrum of $H$ is absolutely continuous outside $\Sigma^D$.
This argument, together with Lemma \ref{L:simple_loop_state} finishes the proof of the statements (3) and (5) of the theorem.

From statement (2), the dispersion relation of $H$ consists of the variety $\lambda=D^{-1}(2F_j(\theta)), j\in \overline{1,3}$,
and collection of flat branches $\lambda\in\Sigma^D$ located at some edges of spectral bands.

According to Lemma \ref{L:KuchPost}, function $D(\lambda)$ is a monotonic homeomorphism from each spectral band of the Hill operator onto $[-2,2]$.
Besides, the ranges of functions $2F_1, 2F_2$ and $2F_3$ are all in $[-2,2]$.
Thus, part of the spectrum $\sigma(H^{per})$ that corresponds to each band of Hill operator $H^{per}$ coincides as a set with the part of the absolutely continuous spectrum $\sigma_{ac}(H)$ that consists of three bands.
In another words, operator $H$ has "three times more" non-flat bands than Hill operator $H^{per}$ does.

The function $D^{-1}$ is multiple-valued, thus producing infinitely many bands for any $2F_j, j=\overline{1,3}$.
Two such spectral bands (pre-images of the same $2F_j(B)$) clearly cannot overlap, because function $D(\lambda)$ is monotonic on each band.
Two spectral bands which are pre-images of $2F_j(B)$ and $2F_i(B)$ for different $i,j\in\overline{1,3}$ also cannot overlap because the ranges of functions $F_j, j=\overline{1,3}$ belong to $[-1,1]$ and do not overlap by Proposition \ref{L:F_property} below.
One should notice that although spectral bands do not overlap, they still can touch each other, and as we will see below this is indeed the case.
This can happen at points $(\theta,\lambda)$ such that $D(\lambda)=\pm 2$ or $\pm 2/3$. This proves the statement (4) of the theorem.

Since the spectra of $H$ and $H^{per}$ coincide as sets, we get the last statement of the theorem.
\end{proof}

\begin{corollary}\label{Cor}\indent
\begin{enumerate}
\item Unless the potential $q_0$ is constant, the spectrum $\sigma(H)$ has at least one gap.
\item For a \emph{generic} smooth potential $q_0$, all possible gaps in $\sigma(H)$ are open.
\end{enumerate}
\end{corollary}
\begin{proof}
The first claim of the Corollary follows from the last statement of Theorem \ref{T:main_graphyne} and Borg's theorem \cite{Borg}. Similarly, the second claim follows from Simon's genericity result \cite{Simon_generic} instead of Borg's theorem.
\end{proof}

Graphene has captured physicists' interest because of its unusual electronic properties.
These properties are caused by occurrence of so-called \emph{conical singularities} or \emph{Dirac points}.
Roughly speaking, Dirac points are points where two spectral bands touch and locally form a cone (also known as \emph{Dirac cone}).
One is interested in conical singularities that are stable under small perturbation of the potential not breaking the symmetry.

In the next theorem, we will take a closer look at the spectral bands of the operator $H$.
Moreover, we will specify all the conical singularities in the Brillouin zone $B$ and describe how the spectral bands behave near these points.

In what follows, we will use the notation
$$\theta_0:=\arccos(-1/3).$$

\begin{theorem} \label{T:conical_singularity}
\indent
\begin{enumerate}
   \item In the free case, i.e., when the potential $q_0$ is equal to zero, the Bloch variety of $H$ has conical singularities at the following points:
      \begin{enumerate}[i)]
          \item $(\theta,\lambda)=(0,0,(2(k+1)\pi)^2)$, at which $D(\lambda)=2$,
          \item $(\theta,\lambda)=(0,\pm\pi,((2k+1)\pi)^2)$, at which $D(\lambda)=-2$,
          \item $(\theta,\lambda)=(\pm\theta_0,0,(\theta_0+2k\pi)^2)$, at which $D(\lambda)=-2/3$,
          \item $(\theta,\lambda)=(\pm\theta_0,\pm\pi,(2\pi-\theta_0+2k\pi)^2)$, at which $D(\lambda)=2/3$,
          for $k=0,1,2,\ldots$
      \end{enumerate}

   \item When we turn on a small potential $q_0\neq 0$, the conical singularities corresponding to $|D(\lambda)|=2$ will generically split into two smooth branches,
   opening a gap.
   The conical singularities that occur when $|D(\lambda)|=2/3$ are stable under small perturbation by a potential $q$ of the type considered.
\end{enumerate}
\end{theorem}
\begin{figure}[ht!]
\includegraphics[scale=0.7]{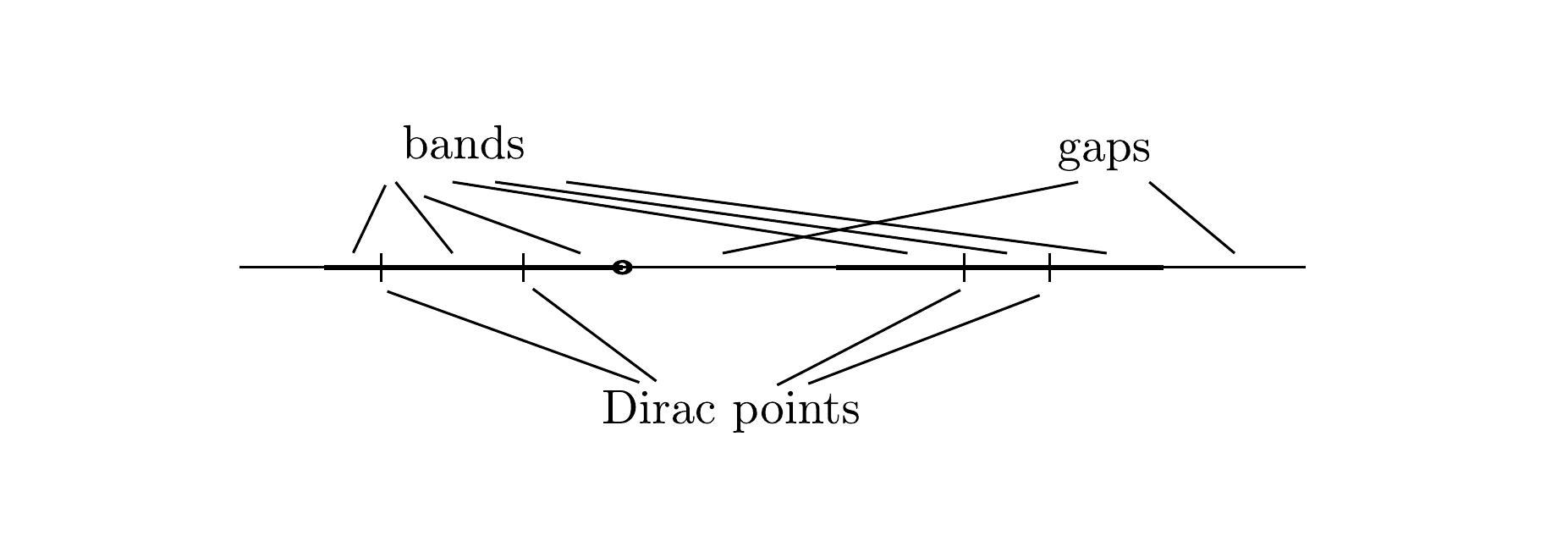}
\caption{The Dirac points are located at the points where three bands of $\sigma(H)$ touch to form a band of $\sigma(H^{per})$.}\label{F:touch}
\end{figure}
Before proving this theorem, we need to explore some properties of function $F(\theta)$.
\begin{prop}\label{L:F_property}
\indent
\begin{enumerate}
   \item Function $F(\theta)$ is real-valued for $\theta$ in the Brillouin zone B.

   \item Let $F(\theta)=(F_1(\theta),F_2(\theta),F_3(\theta))$, where $F_1(\theta)\leq F_2(\theta)\leq F_3(\theta)$ for all $\theta\in B$.
   Then the ranges of functions $F_1,F_2, F_3$ are $[-1,-1/3],[-1/3,1/3]$ and $[1/3,1]$ correspondingly.

   \item Function $F_1$ attains its maximal value at $(\theta_1,0)$ for $\theta_1\in [-\pi,-\theta_0]\cup[\theta_0,\pi]$
   or $(\pm\pi,\theta_2)$ for $\theta_2\in [-\pi,\pi]$
   and minimal value at $\theta=(0,\pm\pi).$

   Function $F_2$ attains its maximal value at $(\theta_1,\pm\pi)$ for $\theta_1\in[-\theta_0,\theta_0]$
   and minimal value at $\theta=(\theta_1,0)$ for  $\theta_1\in[-\theta_0,\theta_0].$

   Function $F_3$ attains its maximal value at $(0,0)$ and minimal value at $(\theta_1,\pm\pi)$
   for $\theta_1\in [-\pi,-\theta_0]\cup[\theta_0,\pi]$ or $(\pm\pi,\theta_2)$ for $\theta_2\in [-\pi,\pi].$

   \item The linear level sets of function $F(\theta)$ inside $B$ are
      \begin{enumerate}[]
          \item $\{(\pm\pi,\theta_2), \theta_2\in [-\pi,\pi]\},$
          \item $\{(\theta_1,\pm\pi),\theta_1\in [-\pi,\pi]\},$
          \item $\{(\theta_1,\pm\pi/2),\theta_1\in [-\pi,\pi]\},$
          \item $\{(\theta_1,0),\theta_1\in [-\pi,\pi]\}$.
      \end{enumerate}
Function $F(\theta)$ does not have any flat branches.
\end{enumerate}
\end{prop}
Proof of this proposition will be given in Section \ref{S:Prop}.

\emph{Proof of Theorem \ref{T:conical_singularity}}

1. In the free case, $D(\lambda)=2\cos\sqrt{\lambda}$ (see, e.g., \cite{Eastham}), so we have
\begin{equation}\label{E:disp_freecase}
\cos\sqrt{\lambda}=F_j(\theta), j=1,2,3.
\end{equation}

As it was proven in Theorem \ref{T:main_graphyne}, spectral bands do not overlap.
It is still possible that these bands touch each other at their edges.
We will study now whether this indeed happens and prove that at these points the spectral bands have conical form.

According to (\ref{E:disp_freecase}), all non-flat spectral bands are
$$\lambda_{6k+j}=(\arccos(F_j(\theta))+2k\pi)^2, \lambda_{6k+3+j}=(2\pi-\arccos(F_j(\theta))+2k\pi)^2,$$
for all $j=\overline{1,3}, k=0,1,2,\ldots$

Thus for $k=0,1,2,\ldots$, we have
\begin{enumerate}[i)]
 \item Bands $\lambda_{6k+4}, \lambda_{6k+7}$ touch each other at $(0,0,(2(k+1)\pi)^2)$, for which $D(\lambda)=2$.

 \item Bands $\lambda_{6k+3}, \lambda_{6k+6}$ touch each other at $(0,\pm\pi,((2k+1)\pi)^2)$, for which $D(\lambda)=-2$.

 \item Bands $\lambda_{6k+2}, \lambda_{6k+3}$ touch each other at $(\pm\theta_0,0,(\theta_0+2k\pi)^2)$ while bands
 $\lambda_{6k+5}, \lambda_{6k+6}$ touch each other at $(\pm\theta_0,0,(-\theta_0+(2k+2)\pi)^2)$.
 At these points $D(\lambda)=-2/3$.

 \item Bands $\lambda_{6k+1}, \lambda_{6k+2}$ touch each other at $(\pm\theta_0,\pm\pi,(\pi-\theta_0+2k\pi)^2)$
 while $\lambda_{6k+4}, \lambda_{6k+5}$ touch each other at $(\pm\theta_0,\pm\pi,(\pi+\theta_0+2k\pi)^2)$.
 At these points $D(\lambda)=2/3$.
\end{enumerate}

Let us now look at the structure near the touching points.
One needs to deal with each case above separately.
Since the argument we use for i) and iii) are similar to those needed for ii) and iv), for the sake of brevity, we will consider only cases i) and iii).

i) Let $(\theta,\lambda)=(0,0,\lambda_0)$ where $\lambda_0=(2(k+1)\pi)^2, k\in\mathbb{N}$. Then
\begin{equation}
\frac{D(\lambda)}{2}= \cos{\sqrt{\lambda}}= 1+a_2(\lambda-\lambda_0)^2+o((\lambda-\lambda_0)^2),
\label{E:D_lamda_freecase}
\end{equation}
for $\lambda\rightarrow\lambda_0$ where $a_2=-1/8\lambda_0< 0,$
\begin{equation}
\cos\theta_1=1-\frac{\theta_1^2}{2}+o(\theta_1^2),\text{ for }\theta_1\rightarrow 0
\label{E:cos1_freecase}
\end{equation}
and
\begin{equation}
\label{E:cos2_freecase}
\cos\theta_2=1-\frac{\theta_2^2}{2}+o(\theta_2^2),\text{ for }\theta_2\rightarrow 0.
\end{equation}

Since $0.5 D(\lambda)$ is a root of the equation (\ref{E:eta_equation}), we have
\begin{equation} \label{E:discriminant_equation}
9\left(\frac{D(\lambda)}{2}\right)^3-\frac{D(\lambda)}{2}=(\cos\theta_1+1)\left(3\frac{D(\lambda)}{2}+\cos\theta_2\right).
\end{equation}

Plugging expressions (\ref{E:D_lamda_freecase}), (\ref{E:cos1_freecase}) and (\ref{E:cos2_freecase}) into (\ref{E:discriminant_equation}) and simplify the expression, one obtains the following formula:
\begin{equation}
\label{E:cone_freecase} A(\lambda-\lambda_0)^2+2\theta_1^2+\theta_2^2=o(\theta_1^2)+o(\theta_2^2)+o((\lambda-\lambda_0)^2),
\end{equation}
for $(\theta,\lambda)\rightarrow(0,0,\lambda_0)$ where $ A=20 a_2< 0$.
Equation (\ref{E:cone_freecase}) shows that the spectral bands touching at the point $(\theta,\lambda)=(0,0,\lambda_0)$ have the conical form.
In other words, the Bloch variety of operator $H$ has conical singularity at $(\theta,\lambda)=(0,0,\lambda_0).$

iii) Let now $(\theta,\lambda)=(\theta_0,0,\lambda_0)$ where $\lambda_0=(\theta_0+2k\pi)^2, k=0,1,2,\ldots$.
We have
$$\frac{D(\lambda)}{2}= -\frac{1}{3}+a_1(\lambda-\lambda_0)+a_2(\lambda-\lambda_0)^2+o((\lambda-\lambda_0)^2),$$
$$\text{ for } \lambda\rightarrow\lambda_0, \text{ where } a_1=0.5D'(\lambda_0),$$

$$ \cos\theta_1=-\frac{1}{3}+b_1(\theta_1-\theta_0)+b_2(\theta_1-\theta_0)^2+o((\theta_1-\theta_0)^2),$$
$$\text{ for } \theta_1\rightarrow\theta_0 \text{ where } b_1=-\sin\theta_0,$$
and
$$\cos\theta_2=1-\frac{\theta_2^2}{2}+o(\theta_2^2), \text{ for } \theta_2\rightarrow 0.$$

Analogously to part i), we substitute these formulas into (\ref{E:discriminant_equation}) to obtain
$$\frac{(A(\lambda-\lambda_0)+b_1(\theta_1-\theta_0))^2}{4}-\frac{B(\theta_1-\theta_0)^2}{4}-\frac{\theta_2^2}{3}=$$
$$=o(\theta_2^2)+o((\theta_1-\theta_0)^2)+o((\lambda-\lambda_0)^2), $$
for $(\theta,\lambda)\rightarrow (\theta_0,0,\lambda_0)$ where $A=6a_1$ and $B=b_1^2>0$.
Since $D'(\lambda)\neq 0$ if $D(\lambda)\neq\pm 2$ according to Lemma \ref{L:KuchPost}, $D'(\lambda_0)\neq 0$ and so $A=6 a_1=3D'(\lambda_0)\neq 0$.
Thus again two spectral bands which touch at the point $(\theta_0,0,\lambda_0)$ have conical form near that point.
The Bloch variety of operator $H$ has conical singularity at $(\theta_0,0,\lambda_0).$

The same argument applies to $(\theta,\lambda)=(-\theta_0,0,\lambda_0)$.
\begin{figure}[ht!]
\includegraphics[scale=0.21]{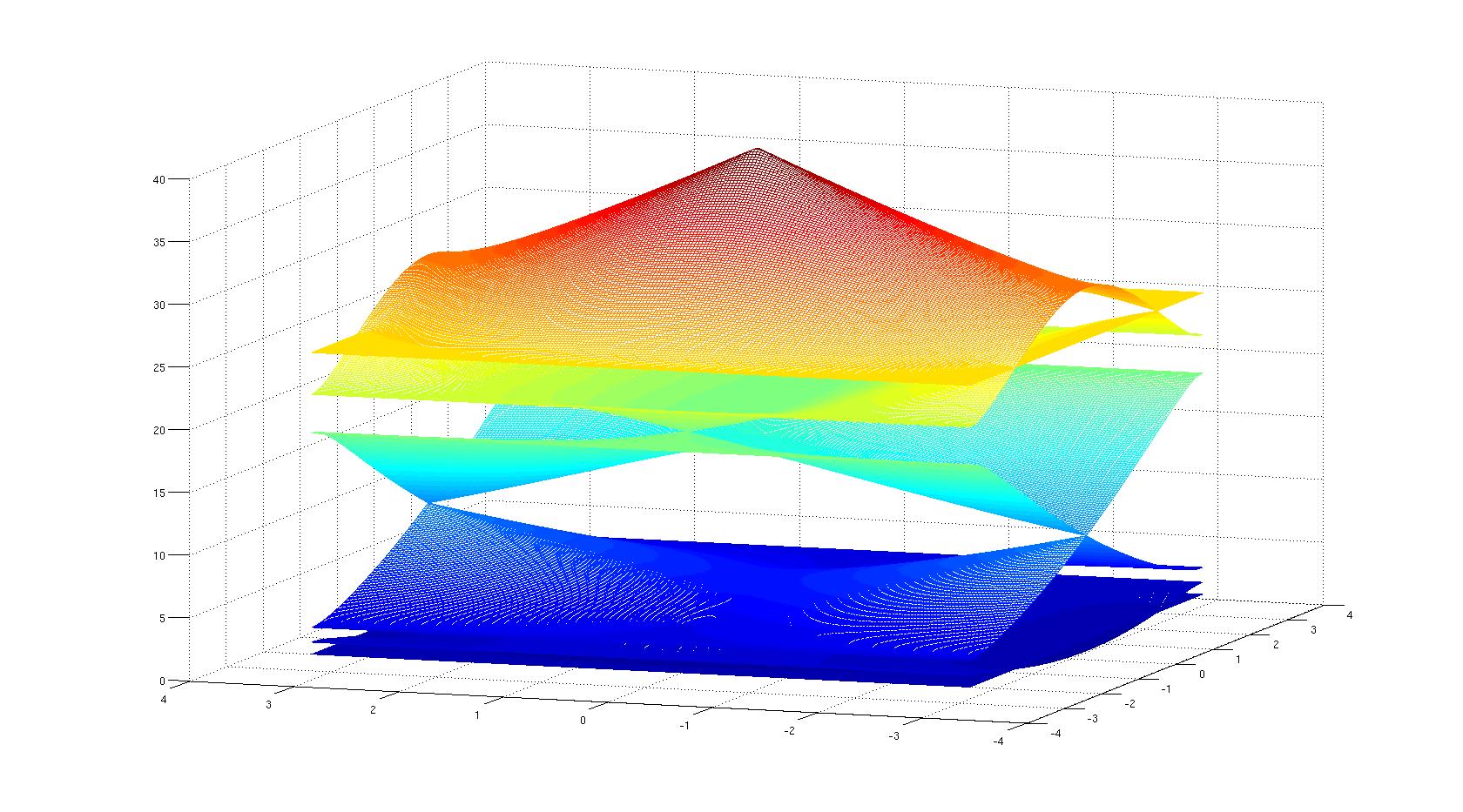}
\caption{Bloch variety of operator $H$ in the free case}\label{F:g}
\end{figure}

2. When we turn on a small potential, since $D'(\lambda_0)\neq 0$ for $\lambda_0=D^{-1}(\pm2/3)$, we can repeat the calculation in part iii),
and find conical singularities at $\theta=(\pm\theta_0,0)$ or $(\pm\theta_0,\pm\pi)$ (at these points $|D(\lambda)|=2/3$).

As it was shown in \cite{Simon_generic}, for almost every $C^{\infty}$ periodic potential $q_0$ on $\mathbb{R}$, all the gaps Hill operator $H^{per}$ open
(at the edges of the gaps $D(\lambda)=\pm 2$).
Since the spectrum $\sigma(H)$ has gaps iff $\sigma(H^{per})$ has gaps, this implies that all conical singularities in the free case with $D(\lambda)=\pm 2$ will generically split into two smooth branches and open a gap when we perturb the potential a little bit.
\qed

\section{Proof of Proposition \ref{L:F_property}}\label{S:Prop}

\begin{proof}

\begin{figure}[b!]
\includegraphics [scale=0.75]{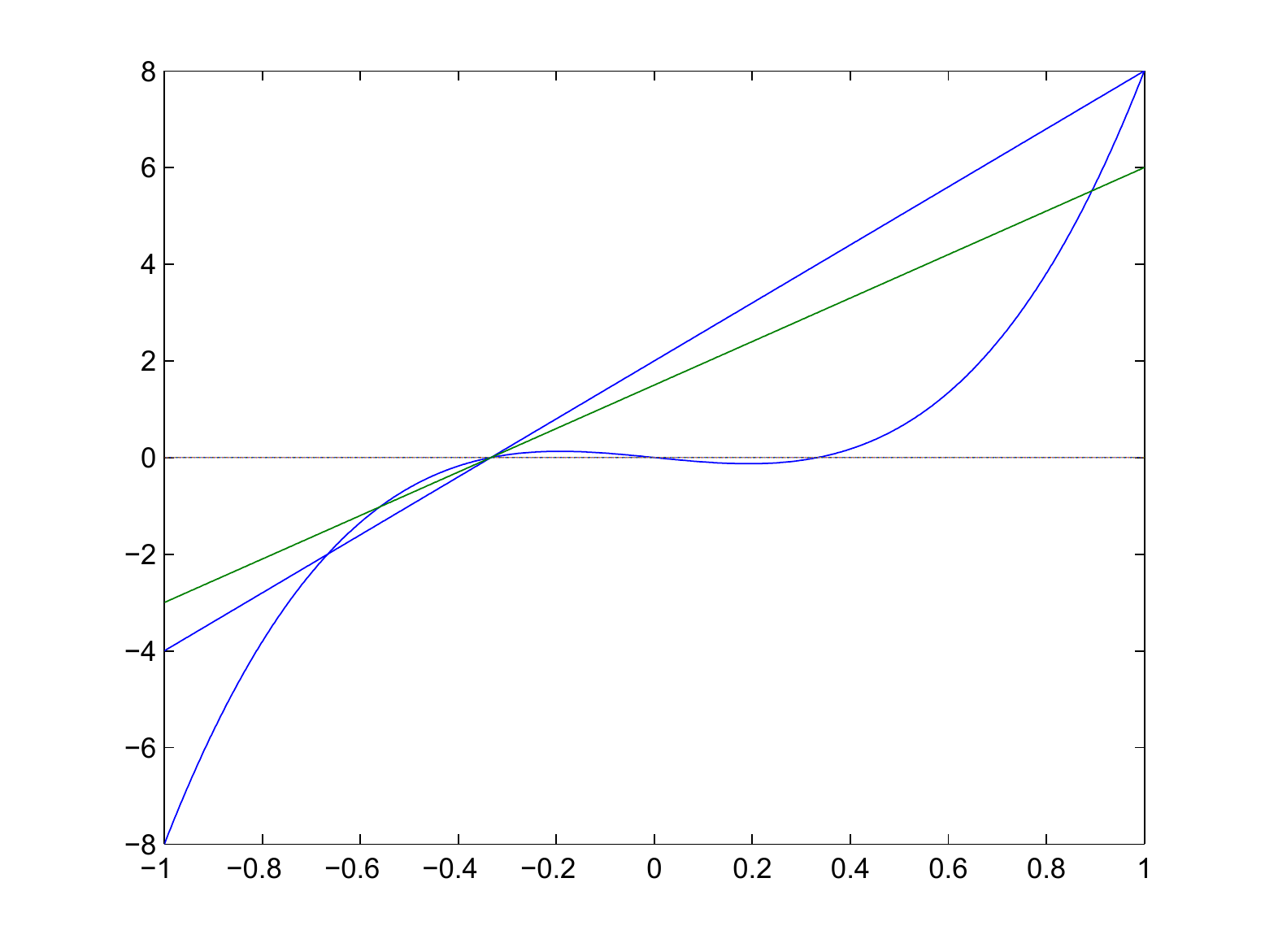}
\caption{The line $l_{1,1}$ lies between $l_{0,1}$ and $l_{2,1}$ and thus has nonempty intersections with $I_1, I_2$ and $I_3$}
\end{figure}
1. We will use the following notations $a:=\cos\theta_1+1$, $b:=\cos\theta_2$, then $a\in[0,2], b\in [-1,1]$.
Equation (\ref{E:eta_equation}) becomes:
\begin{equation}\label{E:x_eta}
9x^3-x=a(3x+b).
\end{equation}

Denote $l_{a,b}$ as the graph of function $y=a(3x+b)$ and $I_1, I_2, I_3$ as parts of the graph of function $y=9 x^3-x$ restricted to $[-1,-1/3]$, $[-1/3,1/3]$ and $[1/3,1]$ correspondingly, $I:=I_1\cup I_2\cup I_3$. Then $I_1, I_2, I_3$ are all connected.

One can notice that the line $l_{0,1}$ intersects with $I_1, I_2, I_3$ when $x=-1/3,0,1/3$ correspondingly while the line $l_{2,1}$ intersects with $I_1, I_2, I_3$ when $x=-2/3, 1/3, 1$ respectively.
When slope $3a$ of the line $l_{a,1}$ changes from 0 to 6, the line $l_{a,1}$ rotates from the line $l_{0,1}$ to the line $l_{2,1}$ around point $(1,0)$.
Thus, the line $l_{a,1}$ intersects with each of $I_1, I_2, I_3$ for all values $a\in[0,2]$.
Applying the same argument for the line $l_{a,-1}$, we also have that the line $l_{a,-1}$ intersects with each of $I_1, I_2, I_3$ for all $a\in[0,2]$.

Besides, for $a\in[0,2], b\in (-1,1)$, three lines $l_{a,1},l_{a,b}$ and $l_{a,-1}$ are parallel and $l_{a,b}$ lies between the other two.
Both lines $l_{a,1}$ and $l_{a,-1}$ intersect with each of $I_1, I_2, I_3$,
and so the line $l_{a,b}$ also intersects with each of $I_1, I_2, I_3$.
This means equation (\ref{E:x_eta}) has three real roots for all $a\in [0,2]$, $b\in [-1,1]$.
Thus all roots of equation (\ref{E:eta_equation}) are real and so function $F(\theta)$ has only real values for all $\theta\in B$.

\begin{figure}[h!]
\includegraphics [scale=0.75]{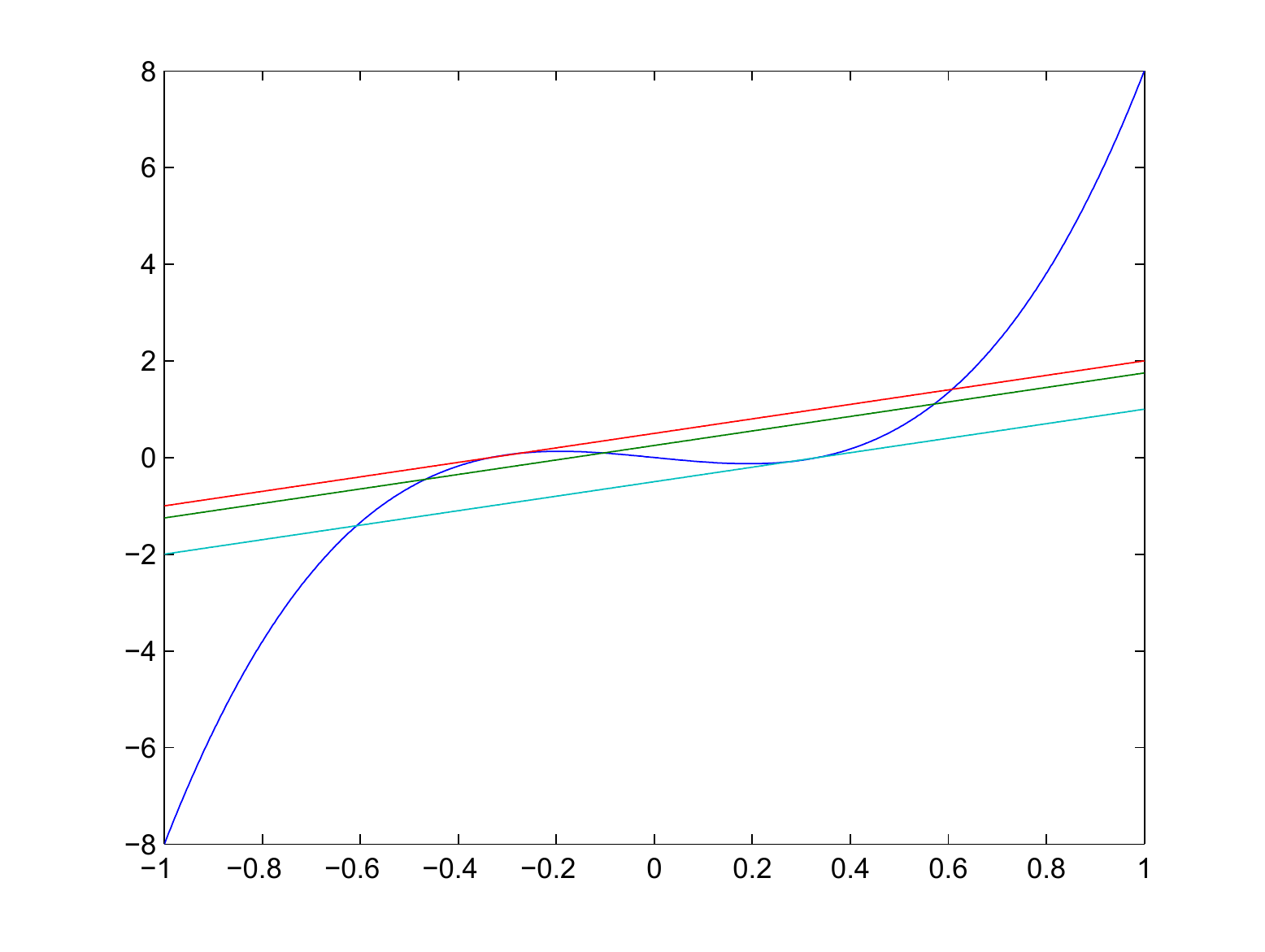}
\caption{The line $l_{1/2,1/2}$ is parallel to $l_{1/2,1}$ and $l_{1/2,-1},$ thus intersects with $I_1, I_2, I_3$}\label{}
\end{figure}
\begin{figure}[tbp!]
\includegraphics[trim=1.5in 3in 1.5in 3in,clip,width=4in]{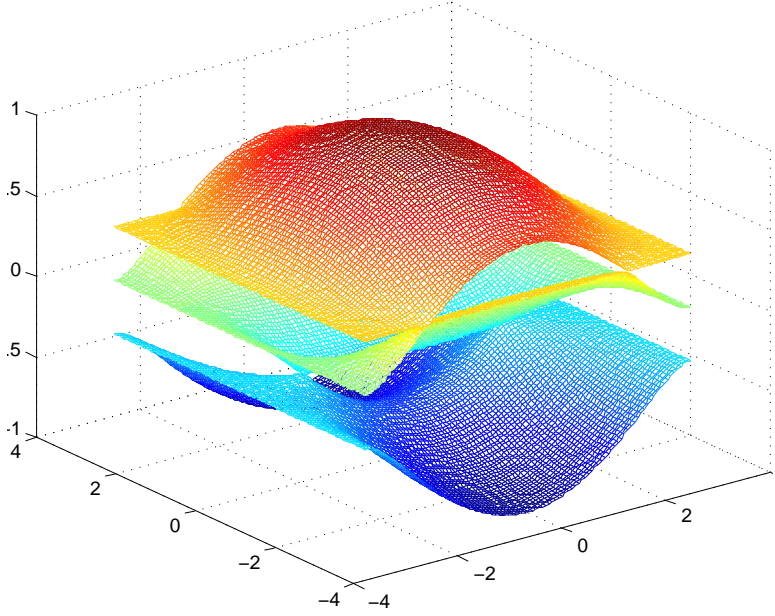}
\caption{Graph of root function for equation (\ref{E:eta_equation})}\label{}
\end{figure}
2. For each $a\in [0,2]$, $b\in [-1,1]$, the line $l_{a,b}$ intersects with each of $I_1, I_2, I_3$.
By our notation, $I_1, I_2, I_3$ are parts of the graph of function $y=9 x^3-x$ restricted to $[-1,-1/3]$, $[-1/3,1/3]$ and $[1/3,1]$ correspondingly.
Thus, the ranges of functions $F_1, F_2$ and $F_3$ are $[-1,-1/3], [-1/3,1/3]$ and $[1/3,1]$ respectively.

3. In what follows, we will find out when function $F_3$ attains its maximal and minimal values.
The similar argument applies for functions $F_1$ and $F_2$.

From part 2 we know that maximal value of $F_3$ is $1$ and its minimal value is $1/3$.

Plugging $x=1$ into equation (\ref{E:x_eta}) we have $a(3+b)=8$, which occurs only when $a=2,b=1$, i.e. $\theta=(0,0)$.
So function $F_1$ attains its maximum at $(0,0)$.

Similarly, we plug $x=1/3$ into the equation (\ref{E:x_eta}) to obtain $a(1+b)=0$, i.e. either $a=0$ or $b=-1$.
Now for each case when $a=0$ or $b=-1$, solve equation (\ref{E:x_eta}), we will see that the biggest root of equation (\ref{E:x_eta}) is equal to $1/3$ when $a=0, b\in[-1,1]$ or $a\in [0,2/3], b=-1$.
This means function $F_3$ attains its minimum at $(\pm\pi,\theta_2)$ for $\theta_2\in [-\pi,\pi]$ or $(\theta_1,\pm\pi)$ for $\theta_1\in [-\pi,-\theta_0]\cup[\theta_0,\pi]$.

4. Let us denote the linear level set of function $F(\theta)$ as $L$ (if such a set exists).
$$L:=\{(\theta_1,\theta_2)\in B|p_1^0\theta_1+p_2^0\theta_2=2k_0\pi\}, p_0=(p_1^0,p_2^0)\in\mathbb{Z}^2\backslash\{(0,0)\},k_0\in\mathbb{Z}.$$
For all $\theta$ belonging to $L$, equation (\ref{E:eta_equation}) has (at least) a constant solution, namely $c$.
Then
\begin{equation}\label{E:level_set}
9c^3-c=(\cos\theta_1+1)(3c+\cos\theta_2), \text{ for all } (\theta_1,\theta_2)\in L.
\end{equation}

If $p_2^0=0$, then the linear level set $L=\{(2k_0\pi/{p_1^0},\theta_2), \theta_2\in[-\pi,\pi]\}$.
Since all values in (\ref{E:level_set}) are constant except $\theta_2$ changing from $-\pi$ to $\pi$,
the expression (\ref{E:level_set}) is true only if $\cos\theta_1+1=0$.
This would mean $\theta_1=\pm\pi$, i.e. $L=\{(\pm\pi,\theta_2),\theta_2\in[-\pi,\pi]\}$.

In case $p_2^0\neq 0$, the linear level set $L$ can be rewritten as
$$L=\left\{(\theta_1,\theta_2)\in B \big| \theta_2=\frac{-p_1^0\theta_1+2k_0\pi}{p_2^0}\right\}$$
and so (\ref{E:level_set}) becomes
$$9c^3-c=(\cos\theta_1+1)\left(3c+\cos\frac{-p_1^0\theta_1+2k_0\pi}{p_2^0}\right), \text{ for all }\theta_1\in [-\pi,\pi].$$
Since $c$ is a constant and $\theta_1$ runs from $-\pi$ to $\pi$, we have
$$9c^3-c=(\cos\pi+1)\left(3c+\cos\frac{-p_1^0\pi+2k_0\pi}{p_2^0}\right)=0.$$
Thus by solving the equation $9c^3-c=0$, we conclude that constant $c$ can be $0, 1/3$ or $-1/3$.
Plugging each value of $c$ into (\ref{E:level_set}), we can get all the linear level sets of $F(\theta)$ as below:
\begin{enumerate}[]
   \item $\{(\pm\pi,\theta_2), \theta_2\in [-\pi,\pi]\},$
   \item $\{(\theta_1,\pm\pi),\theta_1\in [-\pi,\pi]\},$
   \item $\{(\theta_1,\pm\pi/2),\theta_1\in [-\pi,\pi]\},$
   \item $\{(\theta_1,0),\theta_1\in [-\pi,\pi]\}$.
\end{enumerate}
As a consequence, function $F(\theta)$ does not have any flat branches.
\end{proof}

\section{Final remarks and acknowledgments}
\begin{enumerate}

\item The lithographite structure $G$ is not completely flat (due to the presence of four bonds converging at some vertices) \cite{BucCas}. This, however, does not change the quantum network model that we study. Additionally, this structure is the least stable of the 14 configurations studied in \cite{Ivanovskii}. It is, however, the easiest to study among graphynes. Indeed, complexity of the analysis grows with the number of atoms contained in a fundamental domain. This makes graphene the simplest (with just two atoms in an appropriately chosen fundamental domain) and the structure $G$ of this work the next simplest, with three atoms.

\item The structure $G$ has much less than a honeycomb symmetry, which does not prevent it from displaying Dirac cones. This happens also in various other graphyne structures (e.g. \cite{Ivanovskii}).

\item As Figure \ref{F:locone} shows, the Dirac cones in this structure are highly anisotropic, indicating a very directional conductance. This directionality effect, very sharply presented in the structure under consideration, has been noticed for several graphyne structures (e.g., \cite{Gorling}). This is one of the features making graphynes fascinating.
\begin{center}
\begin{figure}[ht!]
\includegraphics[scale=0.4]{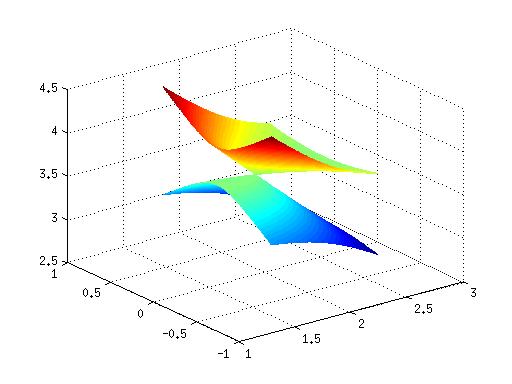}
\caption{The local view of one of the stable Dirac cones. One notices that conductivity will be suppressed in one direction.}\label{F:locone}
\end{figure}
\end{center}

\item In contrast to the density functional calculations of \cite{Ivanovskii} for the structure $G$ of this work, our results show no band overlap. Even though we and the authors of \cite{Ivanovskii} study different approximate models, they have the same geometry and thus are expected to show coherent fetures.

    It is interesting to observe that computations for some other graphyne structures (e.g., \# 10) in \cite{Ivanovskii} do correspond well to such effects arising in the quantum network models.

\item The bound states arising in quantum graph models of graphene and graphyne (which cannot arise for the full dimensional periodic Schr\"odinger equations, e.g. \cite{Kuch_Floquet_book,Reed_Simon_4}), should probably still suggest existence of some rather flat bands/ strong resonances (compare with the photonic crystal situation in \cite{KuKu}). This is confirmed by computations in \cite{Ivanovskii}, except for the case of the lithographite structure $G$, which shows such flat bands in the quantum graph model, while these are absent in the results of  \cite{Ivanovskii}.

\item Due to the small symmetry group of the structure, the class of all invariant potentials is somewhat wider than the one we considered. Namely, it comes from two potentials: $q_0$ on $[0,1]$ and $q_1$ (on $[-1, 1]$)\footnote{In this work, the potential $q_1$ is just the concatenation of two copies of $q_0$.}. It is somewhat harder technically to study this more general class of potentials, but the authors plan to address this issue in a future article.

\item The Corollary \ref{Cor} and its analog for graphene (see \cite{Kuch_Post}) suggest that any non-trivial ``obstacle'' (potential) along the edges opens spectral gaps. For instance, the ``supergraphene'' structure \# 10 in \cite{Ivanovskii}, which differs from the standard graphene by presence of extra two atoms along each edge, is expected to and indeed does show gap opening (as well as very flat bands at some gap edges, similar to the ones in our results).

\item Besides studying more general invariant potentials, the aim of a future work is to consider spectra of nanotubes folded from the graphyne structure of this work. This, in particular, is the reason of the presence of the statement (4) of Proposition \ref{L:F_property}, which will play significant role there.
\end{enumerate}
The authors express their gratitude to the referees for their substantial remarks and additional references.

\end{document}